\documentclass{article}
\usepackage[utf8]{inputenc}
\usepackage{amsthm}
\usepackage{graphicx,ae,aecompl}
\usepackage{amssymb}
\usepackage{amsmath}
\usepackage{subfigure}
\usepackage{undertilde}
\usepackage{algorithmic}
\usepackage[boxruled,linesnumbered]{algorithm2e}
\usepackage[left=2.5cm,top=2.5cm,right=2.5cm,bottom=2.5cm]{geometry}
\usepackage{makeidx}
\usepackage{float}
\usepackage{xcolor}
\usepackage{scalefnt}
\newtheorem{prop}{Proposition}

\title{Asymmetric prior in wavelet shrinkage}
\author{Alex Rodrigo dos Santos Sousa \\ University of São Paulo, Brazil \\ alex.sousa89@usp.br }

\begin{document}

\numberwithin{equation}{section}
\numberwithin{table}{section}
\numberwithin{figure}{section}

 \maketitle
    \begin{abstract}
        
In bayesian wavelet shrinkage, the already proposed priors to wavelet coefficients are assumed to be symmetric around zero. Although this assumption is reasonable in many applications, it is not general. The present paper proposes the use of an asymmetric shrinkage rule based on the discrete mixture of a point mass function at zero and an asymmetric beta distribution as prior to the wavelet coefficients in a non-parametric regression model. Statistical properties such as bias, variance, classical and bayesian risks of the associated asymmetric rule are provided and performances of the proposed rule are obtained in simulation studies involving artificial asymmetric distributed coefficients and the Donoho-Johnstone test functions. Application in a seismic real dataset is also analyzed. 
    \end{abstract}

\section{Introduction}
Wavelet-based methods have been extensively studied and applied in several areas, such as mathematics, signal and image processing, geophysics, genomics and many others. In statistics, applications of wavelets arise mainly in the areas of non-parametric regression, density estimation, functional data analysis and time series analysis. In non-parametric regression, the focus of this work, an unknown function is expanded as linear combination of wavelet basis and the coefficients of this representation are estimated. The use of wavelets representation is attractive in non-parametric regression due their well localized and sparse wavelet coefficients, i.e, the coefficients are typically non-zero or significant on positions where the function has important characteristics to be recovered, as peaks, cusps and discontinuities for example and are zero or very close to zero on smooth regions of the function. These features of wavelets provide computational and analytical advantages. More about wavelet methods in statistics can be seen in Vidakovic (1999). 

Due the sparsity property of wavelet coefficients, shrinkage and thresholding methods are generally used to estimate them in the wavelet domain by reducing the magnitude of the observed (empirical) coefficients obtained by application of a discrete wavelet transformation on the original data. There are in fact several shrinkage and thresholding techniques available in the literature. The main works in this area are of Donoho (1995a, 1995b), Donoho and Johnstone (1994a, 1994b, 1995), but also Donoho et al. (1995, 1996), Vidakovic (1998), Antoniadis et al. (2002) and Johnstone and Silverman (2005) can be cited. For more details of shrinkage methods, see Vidakovic (1999) and Jansen (2001). 

Bayesian shrinkage methods have also been studied, mainly for the possibility of adding, by means of a prior probabilistic distributions,  prior information about the regression, coefficients and other parameters to be estimated. Specifically in the case of wavelets, information about the degree of sparsity of the coefficient vector, the support of these coefficients (if they are limited), among others can be incorporated into the statistical model of study by means of bayesian procedures. In this sense, the choice of the prior distribution of the wavelet coefficients is extremely important to achieve meaningful results. Several bayesian shrinkage procedures have been proposed in the last years in many statistical fields. Some of them are found in Lian (2011), Berger et al. (2012), Karagiannis et al. (2015), Griffin and Brown (2017) and Torkamani and Sadeghzadeh (2017). Further, priors models in the wavelet domain were proposed since 1990s, as Chipman et all (1997), Abramovich and Benjamini (1996), Abramovich et al. (1998), Vidakovic (1998), Vidakovic and Ruggeri (2001), Angelini and Vidakovic (2004), Johnstone and Silverman (2005), Reményi and Vidakovic (2015), Bhattacharya et al. (2015), Sousa (2020), among others.

One feature of the priors already proposed to the wavelet coefficients is that they are symmetric around zero. Athough these priors have been well suceeded in many real applications, this symmetry assumption is not a general case, i.e, wavelet coefficients can be asymmetrically distributed. In this sense, the proposition of an asymmetric prior distribution could be welcome for better estimation of asymmetrically distributed wavelet coefficients. Moreover, little attention has been given to bounded priors, which can be important to model bounded energy signals denoising, restricted to the proposition of the uniform and Bickel distributions by Angelini and Vidakovic (2004), although bounded energy signals occur in practice. Motivated by these reasons, we propose in this work an asymmetric prior distribution to wavelet coefficients based on a discrete mixture of a point mass function at zero and the beta distribution with support on $[-m,m]$. The novelty of this study is, therefore, the application of an asymmetric shrinkage rule associated with this prior model to estimate asymmetrically distributed wavelet coefficients of the unknown function to be recovered in a non-parametric regression model.    

The use of the asymmetric beta prior is interesting by several reasons. First of all, it has a well known shape flexibility obtained by convenient choices of its hyperparameters. Further, its hyperparameters have direct and easy interpretations in terms of asymmetry (left and right asymmetry choices) and shrinkage level, which are very useful for their elicitations by practitioners. Finally, the shrinkage rule under asymmetric beta prior outperformed, in terms of averaged mean squared error, the considered shrinkage/thresholding methods in our simulations studies, mainly when the coefficients present high asymmetry. 

This paper is organized as follows: Section 2 defines the model and the proposed asymmetric beta prior, Section 3 establishes the associated shrinkage rule, shows statistical properties of the rule, such as variance, bias and risks. Parameter and hyperparametes elicitations are discussed in Section 4. Simulation studies involving artificial asymmetric coefficients and the so called Donoho-Jonhstone test functions to evaluate performances are done in Section 5 and application of the proposed shrinkage rule in a real seismic dataset is done in Section 6. The paper finishes with conclusion and final considerations in Section 7.

\section{Statistical model}

Let is consider the non-parametric regression problem of the form
\begin{equation}\label{eq:modeltime}
y_i = f(x_i) + e_i , \qquad i=1,...,n=2^J, J \in \mathbb{N},
\end{equation}
\noindent where $x_i \in [0,1]$, $i=1,...,n$, $f \in \mathbb{L}_2(\mathbb{R})= \{f:\int f^2 < \infty\}$, and $e_i$, $i=1,...,n$, are zero mean independent normal random variables with unknown variance $\sigma^2$. In vector notation, we have
\begin{equation}\label{modeltimevec}
\boldsymbol{y} = \boldsymbol{f} + \boldsymbol{e},
\end{equation}

\noindent where $\boldsymbol{y} = (y_1,...,y_n)'$, $\boldsymbol{f} = (f(x_1),...,f(x_n))'$ and $\boldsymbol{e} = (e_1,...,e_n)'$. To estimate the unknown function $f$, the standard procedure is to apply a discrete wavelet transform (DWT) on \eqref{modeltimevec}, represented by an orthogonal matrix $\boldsymbol{W}$, to obtain the following model in the wavelet domain,
\begin{equation} \label{modelvec}
\boldsymbol{d} = \boldsymbol{\theta} + \boldsymbol{\epsilon},
\end{equation}
where $\boldsymbol{d}=\boldsymbol{W}\boldsymbol{y}$, $\boldsymbol{\theta}=\boldsymbol{W}\boldsymbol{f}$ and $\boldsymbol{\epsilon}=\boldsymbol{W}\boldsymbol{e}$. For a specific coefficient of the vector $\boldsymbol{d}$, we have the simple model 
\begin{equation}\label{model}
d = \theta + \epsilon,
\end{equation}
\noindent where $d$ is the empirical wavelet coefficient, $\theta \in \mathbb{R}$ is the wavelet coefficient to be estimated and $\epsilon \sim N(0,\sigma^2)$ is the normal random error with unknown variance $\sigma^2$. Since the method works coefficient by coefficient, we extract the typical resolution level and location subindices of $d$, $\theta$ and $\epsilon$ for simplicity. Note that, according to the model \eqref{model}, $d|\theta \sim N(\theta,\sigma^2)$ and then, the problem of estimating a function $f$ becomes a normal location parameter estimation problem in the wavelet domain for each coefficient, with posterior estimation of $f$ by the inverse discrete wavelet transform (IDWT), i.e,  $\boldsymbol{\hat{f}} = \boldsymbol{W^{T}}\boldsymbol{\hat{\theta}}$. 

One of the main advantages of expanding an unknown function in wavelet basis is the typical sparsity of the vector of wavelet coefficients $\boldsymbol{\theta}$. In fact, the coefficients are nonzero or significant only in localizations of time domain where the function has features to be recovered, as discontinuities or peaks for example and most of the remaining coefficients are zero or very close to zero. In this sense, a good estimator of $\boldsymbol{\theta}$ should take this sparsity feature into account. Most of the classical procedures usually apply some kind of thresholding policy on the empirical coefficients $\boldsymbol{d}$, i.e, the empirical coefficient $d$ are shrunk to zero if it is less than some threshold value $\lambda$, $\lambda>0$.

For bayesian estimation of $\theta$, it is possible to consider the sparsity of $\boldsymbol{\theta}$ by proposing a discrete mixture of a point mass function at zero and a continuous distribution, attributing most of the weight on the point mass at zero. We propose in this work the following discrete mixture distribution for $\theta$,
\begin{equation}\label{prior}
\pi(\theta;\alpha,a,b,m) = \alpha \delta_{0}(\theta) + (1-\alpha)g(\theta;a,b,m),
\end{equation}
where $\alpha >0$, $\delta_{0}(\cdot)$ is the point mass function at zero and $g(\cdot;a,b,m)$ is the beta density function on $[-m,m]$,
\begin{equation}\label{beta}
g(\theta;a,b,m) = \frac{(\theta+m)^{a-1}(m-\theta)^{b-1}}{(2m)^{a+b-1}B(a,b)}\mathbb{I}_{[-m,m]}(\theta),
\end{equation}
for $a,b,m > 0$, $a \neq b$ (asymmetric case), $B(\cdot,\cdot)$ the standard beta function and $\mathbb{I}_{A}(\cdot)$ the usual indicator function on the set $A$. Thus the proposed prior distribution to the wavelet coefficients has $\alpha, a, b$ and $m$ as hyperparameters to be elicited. We will see later that $\alpha, a$ and $b$ values impact directly on the shrinkage level of the estimator. 

According to the prior models \eqref{prior} and \eqref{beta}, we have that the prior expected value $\mathbb{E}_{\pi}(\theta)$, variance $Var_{\pi}(\theta)$ of $\theta$ are given respectively by
\begin{equation}
\mathbb{E}_{\pi}(\theta) = \frac{m(1-\alpha)(a-b)}{a+b},
\end{equation}
\begin{equation}
Var_{\pi}(\theta) = \frac{(1-\alpha)m^2}{(a+b)^2}\left[\frac{4ab}{a+b+1}+\alpha (a+b)^2\right],
\end{equation}
and if $X$ is a random variable with density \eqref{beta}, then its Pearson skewness coefficient, which is the third standardized moment, is given by
\begin{equation}
Skew(X) = \frac{2(b-a)\sqrt{a+b+1}}{(a+b+2)\sqrt{ab}}.
\end{equation}
First of all, we can note that the proposed prior does not have zero mean wavelet coefficients.  Actually, it only occurs when $a = b$, the symmetric context, which is not considered on this work. Further, the Pearson's skewness coefficient tells us that when $a>b$, $Skew_{\pi}(\theta) < 0$ and the prior is left assymetric. When $a<b$, $Skew_{\pi}(\theta) > 0$ and the prior is right asymmetric. Moreover, the asymmetry increases as $|a-b|$ increases. The symmetry occurs when $a=b$ and then, for this work, we consider only the cases $a\neq b$, i.e, the asymmetric setup. Figure \ref{fig:betadist} shows some beta density functions for $a \neq b$, i.e, for asymmetric cases. Note that there is symmetry around $x=0$ between densities of interchangeable values of parameters $a$ and $b$. We will see in the next section the impact on the shrinkage rules and their statistical properties according to the skewness of the prior distribution of $\theta$.

\begin{figure}[H]
\centering
\includegraphics[scale=0.70]{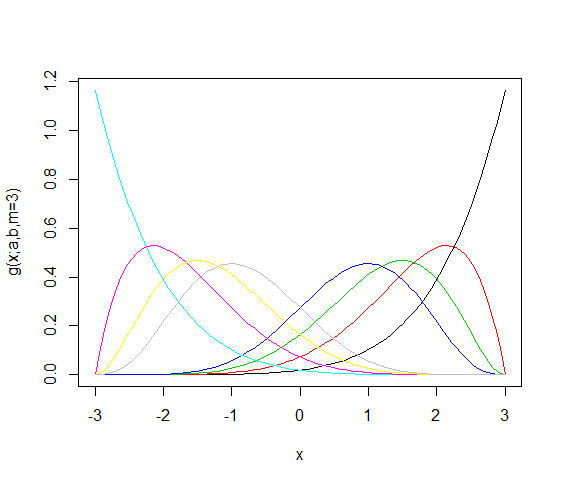}
\caption{Beta density functions \eqref{beta} for $a \neq b$ and $m=3$. When $a>b$, beta density is left assymetric and when $a<b$, it is right asymmetric.}\label{fig:betadist}
\end{figure}
 
\section{The shrinkage rule and statistical properties}
The model \eqref{model}, \eqref{prior} and \eqref{beta} allows us to obtain the bayesian shrinkage rule $\delta(\cdot)$, which is the Bayes estimator of $\theta$ based on the empirical wavelet coefficient $d$. It is well known that under the quadratic loss function $L(\delta,\theta) = (\delta - \theta)^2$, the Bayes estimator is the posterior expected value of $\theta$, i.e, $\delta(d) = \mathbb{E}_{\pi}(\theta|d)$. Proposition \ref{prop1} gives the specific expression of the shrinkage rule under model  
\eqref{model}, a general prior distribution of the form  $\pi(\theta;\alpha,m,\boldsymbol{\tau}) = \alpha \delta_{0}(\theta) + (1-\alpha)g(\theta;\boldsymbol{\tau})$ and for a density function $g$ with support in $[-m,m]$. 

\begin{prop} \label{prop1}
Consider the location parameter $\theta$ estimation problem \eqref{model}. If the prior distribution of $\theta$ is of the form $\pi(\theta;\alpha,m,\boldsymbol{\tau}) = \alpha \delta_{0}(\theta) + (1-\alpha)g(\theta;\boldsymbol{\tau})$, where $g$ is a density function with support in $[-m,m]$ and parameters $\boldsymbol{\tau}$, then the shrinkage rule under the quadratic loss function is given by
\begin{equation}
\delta(d) = \frac{(1-\alpha)\int_{\frac{-m-d}{\sigma}}^{\frac{m-d}{\sigma}}(\sigma u + d)g(\sigma u + d;\boldsymbol{\tau})\phi(u)du}{\alpha \frac{1}{\sigma}\phi(\frac{d}{\sigma})+(1-\alpha)\int_{\frac{-m-d}{\sigma}}^{\frac{m-d}{\sigma}}g(\sigma u + d;\boldsymbol{\tau})\phi(u)du}
\end{equation}
where $\phi(\cdot)$ is the standard normal density function.
\end{prop}

\begin{proof}
If $\mathcal{L}(\cdot \mid \theta)$ is the likelihood function, we have that
\begin{align*}
\delta(d) &= \mathbb{E}_{\pi}(\theta \mid d) \\
          &=\frac{\int_{\Theta}\theta[\alpha\delta_{0}(\theta)+(1-\alpha)g(\theta;\boldsymbol{\tau})]\mathcal{L}(d \mid \theta)d\theta}{\int_{\Theta}[\alpha\delta_{0}(\theta)+(1-\alpha)g(\theta;\boldsymbol{\tau})]\mathcal{L}(d \mid \theta)d\theta} \\
          &= \frac{(1-\alpha)\int_{-m}^{m}\theta g(\theta;\boldsymbol{\tau})\frac{1}{\sqrt{2\pi}}\exp\{-\frac{1}{2}(\frac{d-\theta}{\sigma})^2\}\frac{d\theta}{\sigma}}{\alpha \frac{1}{\sigma\sqrt{2\pi}}\exp\{-\frac{1}{2}(\frac{d}{\sigma})^2\}+(1-\alpha)\int_{-m}^{m}g(\theta;\boldsymbol{\tau})\frac{1}{\sqrt{2\pi}}\exp\{-\frac{1}{2}(\frac{d-\theta}{\sigma})^2\}\frac{d\theta}{\sigma}}\\
          &= \frac{(1-\alpha)\int_{\frac{-m-d}{\sigma}}^{\frac{m-d}{\sigma}}(\sigma u + d)g(\sigma u + d;\boldsymbol{\tau})\phi(u)du}{\alpha \frac{1}{\sigma}\phi(\frac{d}{\sigma})+(1-\alpha)\int_{\frac{-m-d}{\sigma}}^{\frac{m-d}{\sigma}}g(\sigma u + d;\boldsymbol{\tau})\phi(u)du}.\\
\end{align*}
\end{proof}
   
Thus we can apply the Proposition \ref{prop1} to the specific beta density function \eqref{beta} with $\boldsymbol{\tau} = (a,b)'$ to obtain numerically the shrinkage rule of the proposed model. Figure \ref{fig:shrink} presents the shrinkage rules for $\alpha=0.9$, $m=3$ and (a) left asymmetric case $a>b$, for $a=7$ and $b \in \{1,2,3,4,5,6\}$ and (b) right asymmetric case $a<b$, for $b=7$ and $a \in \{1,2,3,4,5,6\}$. First of all, we can observe a symmetric behavior of the rules relative to the origin of the cartesian system for interchangeable choices of $a$ and $b$, i.e, if $\delta_{a,b}(d) = \delta(d)$ for choices of hyperparameters $a$ and $b$, then $\delta_{a,b}(d) = -\delta_{b,a}(-d)$. For example, $\delta_{7,1}(d) = -\delta_{1,7}(-d)$. This feature is a consequence of the symmetry around zero between the respective densites. Thus, we describe the shrinkage rule and its features just for the left asymmetric case ($a>b$), once the same properties occur symmetrically for the right asymmetric distributional context. 

\begin{figure}[H]
\centering
\subfigure[Left asymmetric case $a>b$.]{
\includegraphics[scale=0.55]{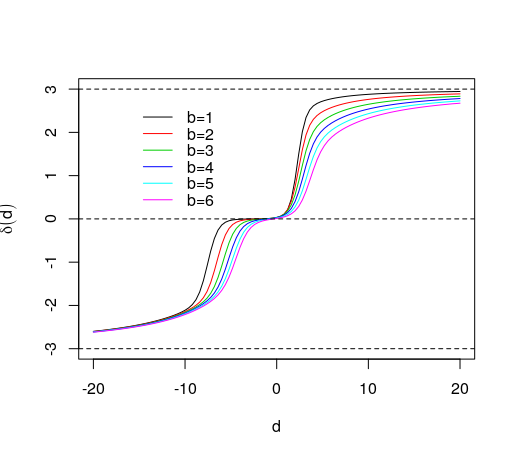}}
\subfigure[Right asymmetric case $a<b$.]{
\includegraphics[scale=0.55]{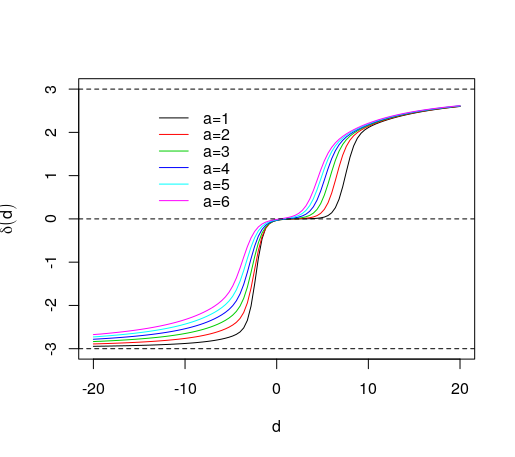}}
\caption{Shrinkage rules under beta prior for $\alpha=0.9$, $m=3$ and (a) left asymmetric case ($a>b$), $a=7$ and (b) right asymmetric case ($a<b$), $b=7$.} \label{fig:shrink}
\end{figure}

Figure \ref{fig:shrink} (a) shows, as expected, that the shrinkage rules under left asymmetric beta prior perform shrinkage asymetrically around zero. In fact, negative empirical coefficients are shrunk more than the positive ones. Moreover, the shrinkage level increases as $|a-b|$ increases, once the interval size of $d$-values that are shrunk to zero is higher. Another property usually taken by bayesian shrinkage rules for bounded wavelet coefficients is that they are also bounded by $[-m,m]$. Since $\theta$ is bounded by $[-m,m]$, empirical coefficients $d$ occur with absolute values higuer than $m$ due noise effect, then they are shrunk to $m$ at most. Figure \ref{fig:bias} (a) presents the impact of the hyperparameter $\alpha$ on shrinkage level of the rules, for $\alpha \in \{.6, .7, .8, .9, .99\}$, $a=7$, $b=3$ and $m=3$. As expected, the rule shrinks more as $\alpha$ increases, since this set more weight to the point mass at zero function in the prior model \eqref{prior}. 

\begin{figure}[H]
\centering
\subfigure[Shrinkage rules for several $\alpha$ values.]{
\includegraphics[scale=0.50]{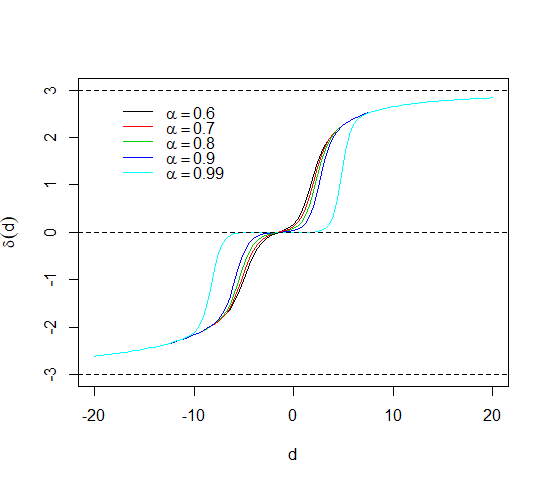}}
\subfigure[Squared bias for several $b$ values.]{
\includegraphics[scale=0.55]{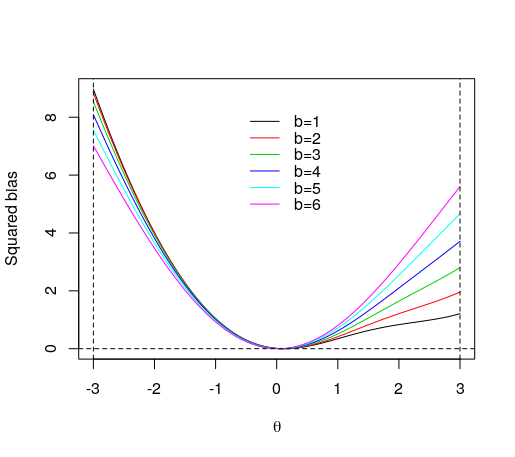}}
\caption{Shrinkage rules under beta prior for $a=7$, $b=3$ and $m=3$ for several values of $\alpha$ (a) and squared bias for shrinkage rules for $\alpha=0.9$, $m=3$ and left asymmetric case ($a>b$) with $a=7$ (b).  } \label{fig:bias}
\end{figure}

Squared bias and variance of the shrinkage rules for left asymmetric case ($a>b$) are provided in Figures \ref{fig:bias} (b) and \ref{fig:varrisk} (a) respectively. The estimators are practically unbiased and achieve minimum variance when $\theta$ is close to zero (but not for $\theta=0$). These features also exist when symmetric priors are assumed for $\theta$. However, these properties behave differently for negative and positive $\theta$ values. The bias increases faster for negative $\theta$ values than for positive ones while the variance increases faster for positive values. Since the shrinkage is stronger for negative values of $d$, it is reasonable the asymmetrical increase of the bias on the negative values direction with the simultaneously decreasing variance toward it. 

\begin{figure}[H]
\centering
\subfigure[Variances]{
\includegraphics[scale=0.55]{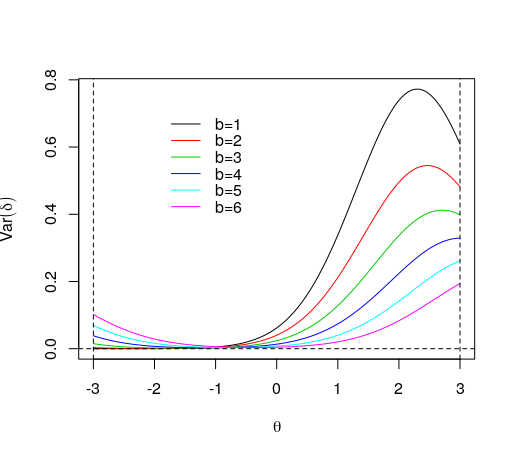}}
\subfigure[Classical risks]{
\includegraphics[scale=0.55]{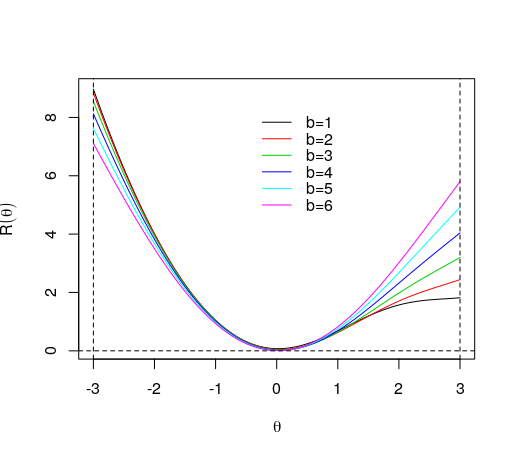}}
\caption{Variances and classical risks for shrinkage rules under beta prior for $\alpha=0.9$, $m=3$ and left asymmetric case ($a>b$) with $a=7$.} \label{fig:varrisk}
\end{figure}

Figure \ref{fig:varrisk} (b) and Table \ref{tab:brisk2} show the classical risks $R_{\delta}(\theta)$ and Bayes risks $r_{\delta}$ respectively for the same rules considered on the plots of squared bias and variance. In fact, the behavior of the classical risks is the same as the squared bias one, i.e, there is a faster increase of the risk for negative $\theta$ values than for positive values, with the minimum risk close to zero. Moreover, we observe that the Bayes risk decreases as the hyperparameter $b$ increases and goes to the symmetric case. Further, Bayes risks for several $\alpha$ values are presented in Table \ref{tab:brisks1}. As expected, the bayesian risk decreases as $\alpha$ increases, once this last one imply a higuer shrinkage and agree with the prior belief of sparsity of the $\boldsymbol{\theta}$ vector.

\begin{table}[!htb]
\centering
\label{my-label}
\begin{tabular}{cccccccc}
\hline
$b$ & 1   & 2 & 3 & 4 & 5 & 6     \\ \hline
$r_{\delta}$ & 0.221 & 0.182 & 0.139 & 0.103 & 0.078 & 0.063 \\ \hline
\end{tabular}
\caption{Bayes risks of the shrinkage rules under beta prior distribution with hyperparameters $\alpha = 0.9$, $m=3$ and $a=7$.}\label{tab:brisk2}
\end{table}

\begin{table}[!htb]
\centering
\label{my-label}
\begin{tabular}{ccccccc}
\hline
$\alpha$ & 0.6   & 0.7 & 0.8 & 0.9 & 0.99     \\ \hline
$r_{\delta}$ & 0.352 & 0.299 & 0.231 & 0.139 & 0.019 \\ \hline
\end{tabular}
\caption{Bayes risks of the shrinkage rules under beta prior distribution with hyperparameters $m=3$, $a=7$ and $b=3$.}\label{tab:brisks1}
\end{table}

\section{Parameters elicitation}

Methods and criteria for determination of the involved parameters and hyperparameters to estimate the coefficients are important in bayesian procedures. In the framework of model \eqref{model}, \eqref{prior} and \eqref{beta}, the choices of the $\sigma$ parameter of the normal random error distribution and the hyperparameters $\alpha$, $m$, $a$ and $b$ of the beta prior distribution of the wavelet coefficient are required. We present the methods and criteria already available in the literature for such choices and used in simulation and application studies and some direction on elicitation of the beta shape parameters.

Based on the fact that much of the noise information present in the data can be obtained on the finer resolution scale, for the robust $\sigma$ estimation, Donoho and Johnstone (1994a) suggest

\begin{equation}\label{eq:sigma}
\hat{\sigma} = \frac{\mbox{median}\{|d_{J-1,k}|:k=0,...,2^{J-1}\}}{0.6745}.
\end{equation}

The hyperparameters $\alpha$ and $m$ are the weight of the point mass function at zero of the proposed prior and the upper value of the beta support respectively. Angelini and Vidakovic (2004) suggest the hyperparameters $\alpha$ and $m$ be dependent on the level of resolution $j$ according to the expressions
\begin{equation}\label{eq:alpha}
\alpha = \alpha(j) = 1 - \frac{1}{(j-J_{0}+1)^\gamma}
\end{equation}
and
\begin{equation}\label{eq:m}
m = m(j) = \max_{k}\{|d_{jk}|\},
\end{equation}
where $J_ 0 \leq j \leq J-1$, $J_0$ is the primary resolution level and $\gamma > 0$. They also suggest that in the absence of additional information, $\gamma = 2$ can be adopted.

Finally, the shape hyperparameters $a$ and $b$ should be chosen according to asymmetry and shrinkage levels criteria. Left asymmetry  imposes $a>b$ and right one says $a<b$. As $|a-b|$ increases, the asymmetry level and shrinkage in the direction of this asymmetry increase. In practice, assuming symmetric around zero noise, as the considered model \eqref{model}, the observed asymmetry of the empirical coefficients $d$ can provide some information about wavelet coefficients asymmetry criteria and be a starting point to elicite it.

\section{Simulation studies}
Two simulation studies were done to evaluate the performance of the proposed shrinkage rule and to compare it with well known shrinkage/thresholding methods. The first one (simulation study 1) had empirical coefficients vector $\boldsymbol{d}$ artificially generated according to the models \eqref{model}, \eqref{prior} and \eqref{beta} and the second one (simulation study 2) involved Donoho-Johnstone test functions, which are extremely applied in the literature to compare wavelet-based methods. 

In both simulation studies, our proposed shrinkage rule performance were compared with soft thresholding  with threshold parameter chosen according to the following policies: universal thresholding (UNIV) proposed by Donoho and Johnstone (1994b), false discovery rate (FDR) proposed by Abramovich and Benjamini (1996), cross validation (CV) of Nason (1996) and Stein unbiased risk estimator (SURE) of Donoho and Johnstone (1995). Moreover, we also compared with bayesian shrinkage methods: bayesian adaptive multiresolution shrinker (BAMS) of Vidakovic and Ruggeri (2001) and large posterior mode (LPM) of Cutillo et al. (2008).

We used the mean squared error (MSE), $MSE = \frac{1}{n} \sum_{i=1}^{n}[{\hat f(x_i)} - f(x_i)]^2$ as performance measure of the shrinkage rules on each run of the simulation. For each function, the simulation was repeated $M$ times and a comparison measure, the average of the obtained MSEs, $AMSE = \frac{1}{M} \sum_{j=1}^{M}MSE_j$, was calculated. Thus, the best method in terms of averaged mean squared error is the one with the smallest value of AMSE.  

In both simulation studies, normal random noise vectors were generated according to three signal to noise ratio values (SNR), 3, 6 and 9, two sample sizes were considered, $n=512$ and $2048$ and wavelet basis Daubechies with eight null moments (Daub8) was applied.

\subsection{Simulation study 1}

To evaluate the performance of the proposed shrinkage rule in asymmetric distributed wavelet coefficients, we generated the wavelet coefficients $\boldsymbol{\theta}$ according to the models \eqref{prior} and \eqref{beta} in two context of asymmetry controlled by the choice of the hyperparameters. The first one for $a=3$, $b=7$, $m=10$ and $\alpha=0.9$, which give $\mathbb{E}_{\pi}(\theta) = -0.40$, $Var_{\pi}(\theta) = 2.20$ and $Skew(X) = 0.48$. Thus, this context has a moderate asymmetry and a large amount of null coefficients due the impact of the weight $\alpha$ of the point mass function at zero. The second context one for more extreme choices, $a=1$, $b=20$, $m=30$ and $\alpha=0.6$, which give $\mathbb{E}_{\pi}(\theta) =10.86$, $Var_{\pi}(\theta) = 179.78$ and $Skew(X) = 1.73$. This last one introduces more asymmetry and less null coefficients. Tables \ref{tab:sim11} and \ref{tab:sim12} show the AMSEs for first and second contexts respectively obtained for $M=1000$ simulation runs in each scenario of sample size and SNR. 

\begin{table}[H]
\scalefont{0.5}
\centering
\label{my-label}
\begin{tabular}{|c|c|c|c|c|||c|c|c|c|c|}
\hline
 
 n & Method & SNR=3 & SNR=6 & SNR=9 & n & Method & SNR=3 &SNR=6 &SNR=9  \\ \hline \hline
512&	UNIV	&	 0.439	&	0.119	&	0.054	& 2048	&	UNIV	&	0.310	&	0.087	&	0.040 \\
&	FDR	&	 0.445 	&	0.129	&	0.067	&	&	FDR	&	0.237	&	0.065	&	0.030 \\
&	CV	&	 0.634 	&	0.177	&	0.082	&	&	CV	&	0.404	&	0.116	&	0.054 \\
&	SURE	&	 0.122 	&	0.031	&	0.014	&	&	SURE	&	0.078	&	0.020	&	0.009 \\
&	BAMS	&	 0.402 	&	0.375	&	0.371	&	&	BAMS	&	0.237	&	0.218	&	0.215 \\
&	LPM	&	 0.288 	&	0.072	&	0.032	&	&	LPM	&	0.240	&	0.060	&	0.026 \\
&	BETA	& \textbf{0.058} 	&	\textbf{0.012}	&	\textbf{0.004}	&	&	BETA	&	\textbf{0.037}	&	\textbf{0.008}	&	\textbf{0.003} \\ \hline

\end{tabular}
\caption{AMSE of the shrinkage/thresholding rules in the simulation study for the empirical wavelet coefficients vector artificially generated according to the models \eqref{model}, \eqref{prior} and \eqref{beta} for $\alpha = 0.9$, $a=3$, $b=7$ and $m=10$.}\label{tab:sim11}
\end{table}

\begin{table}[H]
\scalefont{0.5}
\centering
\label{my-label}
\begin{tabular}{|c|c|c|c|c|||c|c|c|c|c|}
\hline
 
 n & Method & SNR=3 & SNR=6 & SNR=9 & n & Method & SNR=3 &SNR=6 &SNR=9  \\ \hline \hline
512&	UNIV	&	8.316	&	2.082	&	0.922	&2048	&	UNIV	&	12.204	&	3.058	&	0.922 \\
&	FDR	&	14.350	&	27.230	&	26.385	&	&	FDR	&	9.333	&	2.838	&	26.385 \\
&	CV	&	11.456	&	2.854	&	1.261	&	&	CV	&	17.047	&	4.279	&	1.261 \\
&	SURE	&	2.345	&	0.589	&	0.260	&	&	SURE	&	2.680	&	0.672	&	0.260 \\
&	BAMS	&	2.097	&	0.401	&	0.256	&	&	BAMS	&	2.425	&	0.449	&	0.256 \\
&	LPM	&	6.933	&	1.733	&	0.770	&	&	LPM	&	7.464	&	1.867	&	0.770 \\
&	BETA	& \textbf{0.307}	&	\textbf{0.111}	&\textbf{0.057}	&	&	BETA	&	\textbf{0.351}	&	\textbf{0.120}	&	\textbf{0.057} \\ \hline

\end{tabular}
\caption{AMSE of the shrinkage/thresholding rules in the simulation study for the empirical wavelet coefficients vector artificially generated according to the models \eqref{model}, \eqref{prior} and \eqref{beta} for $\alpha = 0.6$, $a=1$, $b=20$ and $m=30$.}\label{tab:sim12}
\end{table}

We can observe an excelent performance of the shrinkage rule under asymmetric prior in both contexts. In fact, our rule had the best performance in terms of AMSE in all the considered scenarios of sample size and SNR and both the contexts of hyperparameters choice. We emphasize the second context, with strong asymmetry and less sparsity degree of the coefficients, where the difference in performance of the rule with the others are significant. It suggests that when the wavelet coefficients are asymmetrically distributed, our proposed rule should be considered as the shrinker to be applied.

Another interesting observed feature is the good performance of the rule for low SNR values. When SNR=3, the beta rule had significant difference against the comparison methods, i.e, the shrinker can be well succeeded even for presence with high noise level in the data, which is a desirable for shrinkers.

Boxplots of MSEs obtained for the shrinkage/thresholding rules in both contexts for $n=512$ and SNR=3 are presented in Figures \ref{fig:bp1}(a) and (b). 

\begin{figure}[H]
\centering
\subfigure[Context 1]{
\includegraphics[scale=0.45]{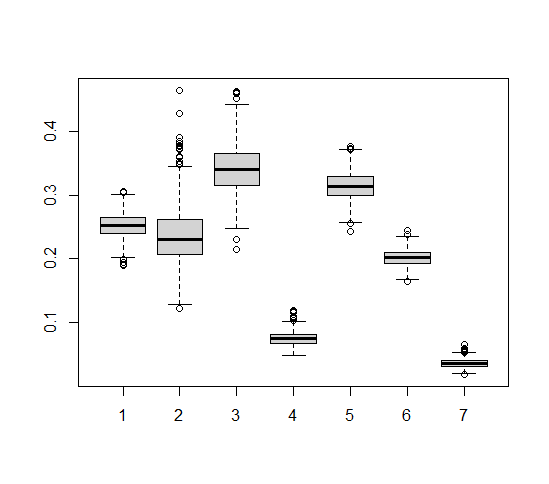}}
\subfigure[Context 2]{
\includegraphics[scale=0.45]{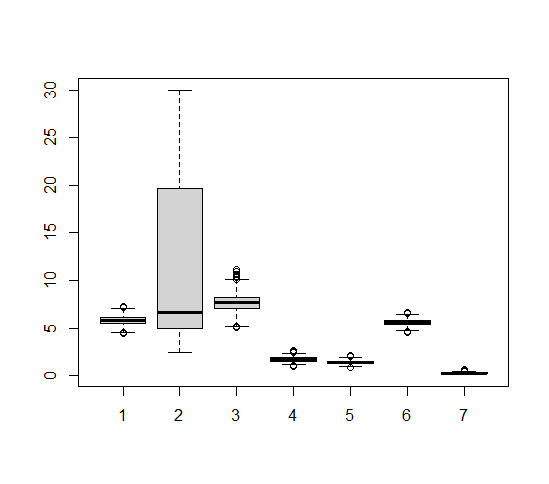}}
\caption{Boxplots of MSEs of shrinkage/thresholding rules in simulation study 1 involving artificial asymmetric distributed wavelet coefficients for $n=512$ and SNR=3 according to the models \eqref{model}, \eqref{prior} and \eqref{beta} for $\alpha = 0.9$, $a=3$, $b=7$ and $m=10$ (a) and $\alpha = 0.6$, $a=1$, $b=20$ and $m=30$ (b). The associated rules are: 1-Univ, 2-FDR, 3-CV, 4-SURE, 5-BAMS, 6-LPM and 7-Beta shrinkage rule.} \label{fig:bp1}
\end{figure}

\subsection{Simulation study 2}

To conclude our simulation studies, we evaluate the performance of the proposed shrinkage rule in the four Donoho-Jonhstone test functions called Bumps, Blocks, Doppler and Heavisine. Figure \ref{fig:DJ}(a) shows these extremely used functions in statistical wavelet research, once each of them has interesting features such as discontinuities, spikes and oscillations that are important to be captured in curve estimation procedures by wavelet modelling.

\begin{figure}[H]
\centering
\subfigure[Donoho-Johnstone test functions.]{
\includegraphics[scale=0.45]{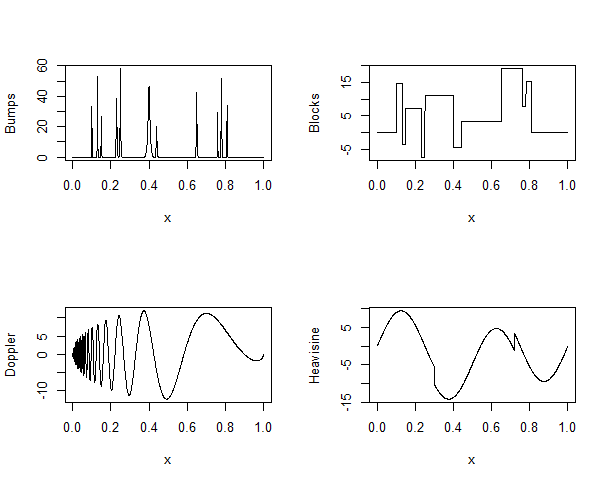}}
\subfigure[Fitted curves.]{
\includegraphics[scale=0.45]{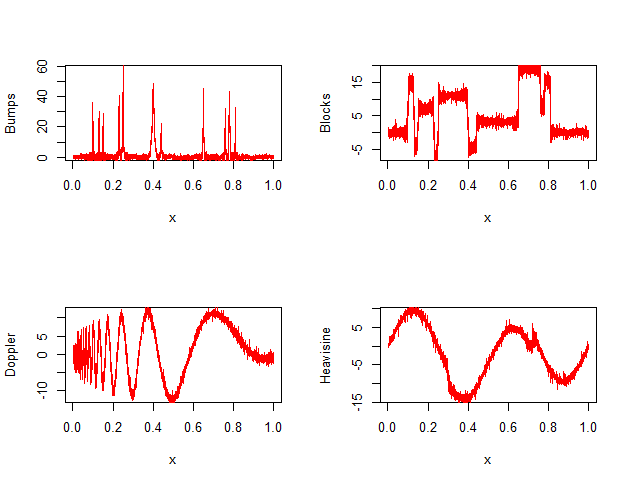}}
\caption{Donoho-Johnstone test functions used as underlying signals in simulation study 2 (a) and fitted curves obtained after application of asymmetric shrinkage rule under beta prior for $n=512$ and SNR=6.} \label{fig:DJ}
\end{figure}

In fact, the four functions present asymetrically distributed wavelet coefficients with different levels. Bumps and Doppler functions have the highest levels of wavelet coefficients asymmetry and the Heavisine, the shortest one, i.e, its coefficients are almost symmetric around zero. These previous knowledge allows us to explain the different performances of our asymmetric shrinker among the functions. 

In each function, data were generated according to the addition of normal random noise with the same scenarios of sample size and SNR of simulation study 1. The AMSEs obtained for each rule are in Tables 
\ref{tab:sim21} and \ref{tab:sim22}, for $M=500$ simulation runs for each scenario.

\begin{table}[H]
\scalefont{0.5}
\centering
\label{my-label}
\begin{tabular}{|c|c|c|c|c|c|||c|c|c|c|c|c|}
\hline
 
Signal & n & Method & SNR=3 & SNR=6 & SNR=9 & Signal & n & Method & SNR=3 &SNR=6 &SNR=9  \\ \hline \hline
Bumps&512	&	Univ	&	11.122	&	3.882	&	2.014	&Blocks	&512	&	Univ	&	6.940	&	2.846	&	1.524	\\
&	&	FDR	&	9.313	&	3.324	&	1.774	&	&	&	FDR	&	5.911	&	2.244	&	1.164	\\
&	&	CV	&	11.444	&	7.665	&	4.533	&	&	&	CV	&	2.575	&	1.000	&	0.673	\\
&	&	SURE	&	3.656	&	1.159	&	0.578	&	&	&	SURE	&	2.827	&\textbf{0.904}	&	0.448	\\
&	&	BAMS	& \textbf{2.833}	&	1.355	&	1.160	&	&	&	BAMS	&\textbf{2.469}	&	1.163	&	1.018	\\
&	&	LPM	&	5.441	&	1.363	&	0.606	&	&	&	LPM	&	5.469	&	1.363	&	0.605	\\
&	&	Beta	&	2.968	& \textbf{1.075}	&\textbf{0.548}	&	&	&	Beta	&	2.834	&	0.950	&\textbf{0.416}	\\ \hline

&2048	&	Univ	&	5.050	&	1.769	&	0.935	&	&2048	&	Univ	&	3.417	&	1.376	&	0.757	\\
&	&	FDR	&	3.582	&	1.173	&	0.602	&	&	&	FDR	&	2.688	&	0.967	&	0.513	\\
&	&	CV	&	1.610	&	0.578	&	0.381	&	&	&	CV	&	1.307	&	0.444	&	0.248	\\
&	&	SURE	&	1.651	& 0.510	&\textbf{0.250}	&	&	&	SURE	&	1.359	&\textbf{0.441}	&	\textbf{0.220}	\\
&	&	BAMS	&	1.635	&	0.573	&	0.482	&	&	&	BAMS	&	1.502	&	0.499	&	0.418	\\
&	&	LPM	&	5.453	&	1.359	&	0.604	&	&	&	LPM	&	5.450	&	1.361	&	0.605	\\
&	&	Beta	& \textbf{1.450}	&\textbf{0.465}	&0.298	&	&	&	Beta	&\textbf{1.306}	&	0.445	&	0.231	\\ \hline

\end{tabular}
\caption{AMSE of the shrinkage/thresholding rules in the simulation study for the empirical wavelet coefficients vector artificially generated with Donoho Jonhstone test functions Bumps and Blocks as underlying signals and additive normal random noise.}\label{tab:sim21}
\end{table}

\begin{table}[H]
\scalefont{0.5}
\centering
\label{my-label}
\begin{tabular}{|c|c|c|c|c|c|||c|c|c|c|c|c|}
\hline
 
Signal & n & Method & SNR=3 & SNR=6 & SNR=9 & Signal & n & Method & SNR=3 &SNR=6 &SNR=9  \\ \hline \hline
Doppler&512	&	Univ	&	2.645	&	1.102	&	0.608	&Heavisine	&512	&	Univ	&	0.570	&	0.348	&	0.237	\\
&	&	FDR	&	2.540	&	0.967	&	0.514	&	&	&	FDR	&	0.594	&	0.368	&	0.228	\\
&	&	CV	&	1.269	&	0.518	&	0.374	&	&	&	CV	&\textbf{0.511}	&\textbf{0.219}	&\textbf{0.122}	\\
&	&	SURE	&	1.318	&\textbf{0.435}	&	0.216	&	&	&	SURE	&	0.573	&	0.360	&	0.249	\\
&	&	BAMS	&	1.527	&	0.537	&	0.461	&	&	&	BAMS	&	1.142	&	0.258	&	0.204	\\
&	&	LPM	&	5.448	&	1.365	&	0.603	&	&	&	LPM	&	5.429	&	1.365	&	0.604	\\
&	&	Beta	&	\textbf{1.303}	&	0.469	&\textbf{0.212}	&	&	&	Beta	&	0.623	&	0.267	&	0.136	\\ \hline

&2048	&	Univ	&	1.155	&	0.457	&	0.257	&	&2048	&	Univ	&	0.358	&	0.193	&	0.123	\\
&	&	FDR	&	1.040	&	0.377	&	0.199	&	&	&	FDR	&	0.389	&	0.186	&	0.111	\\
&	&	CV	&	0.557	&	0.191	&	0.097	&	&	&	CV	&\textbf{0.264}	&\textbf{0.109}	&\textbf{0.060}	\\
&	&	SURE	&	0.573	&	0.194	&\textbf{0.097}	&	&	&	SURE	&	0.360	&	0.197	&	0.111	\\
&	&	BAMS	&	1.084	&	0.211	&	0.160	&	&	&	BAMS	&	0.982	&	0.143	&	0.100	\\
&	&	LPM	&	5.439	&	1.358	&	0.604	&	&	&	LPM	&	5.444	&	1.359	&	0.604	\\
&	&	Beta	&\textbf{0.468}	&\textbf{0.174}	&	0.103	&	&	&	Beta	&	0.548	&	0.795	&	0.326	\\ \hline

\end{tabular}
\caption{AMSE of the shrinkage/thresholding rules in the simulation study for the empirical wavelet coefficients vector artificially generated with Donoho Jonhstone test functions Doppler and Heavisine as underlying signals and additive normal random noise.}\label{tab:sim22}
\end{table}

We observe that the proposed shrinkage rule had great performance in practically all the scenarios. It was the best one in terms of AMSE in Bumps and Doppler functions and beat the comparison rules in some scenarios for Blocks function. Even for Heavisine function, that has the least asymmetric coefficients distribution, which was dominated by CV method, our beta rule had reasonable performance, with AMSE very close to the best methods on almost all the scenarios for this function. This feature should be emphasized: although the asymmetric shrinker had the best performance when applied in empirical wavelet coefficients of functions with significant asymmetrical wavelet coefficients distributions, which was already expected, it also had good results when the wavelet coefficients of the signal are practically symmetric. This provides some sort of flexibility of the proposed shrinker, which is essential for real data applications. Finally, as observed in simulation study 1, we also have good results for low SNR values, such as 3. 

Figures \ref{fig:DJ}(b) and \ref{fig:bpsim2} show fitted curves obtained after denoising by asymmetric shrinkage rule and boxplots of the MSEs of the shrinkage/thresholding rules respectively for $n=512$ and SNR=6. One can note that the fitted curves recover the main features of each signal, as jumps and spikes. Moreover, the MSEs of the proposed rule (rule number 7 at the boxplots) had low variation.

\begin{figure}[H]
\centering
\includegraphics[scale=0.70]{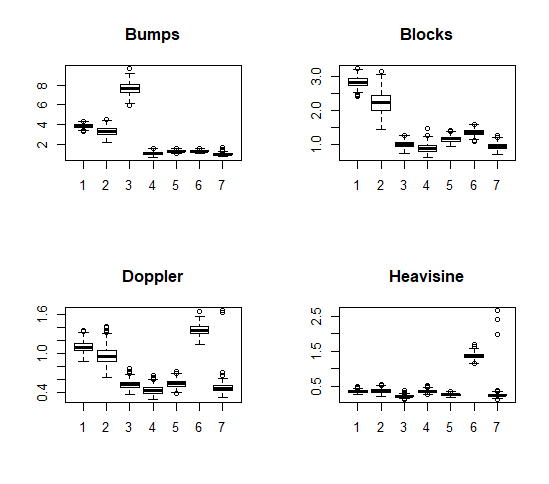}
\caption{Boxplots of MSE of shrinkage/thresholding rules in simulation study 2 involving Donoho-Johnstone test functions for $n=512$ and SNR=6. The associated rules are: 1-Univ, 2-FDR, 3-CV, 4-SURE, 5-BAMS, 6-LPM and 7-Beta shrinkage rule.}\label{fig:bpsim2}
\end{figure}

\section{Application: seismic dataset}

Coso, California, is a geothermal area, with geological strucute tectonically active and its geological and geophysics properties have been studied along the last decades by experts of the area. Once with a seismogram at hand, the experts classify the subsets of the data as primary waves (P-waves), secondary waves (S-waves) and body waves, according to some faetures of the signal, as amplitude, velocity and other geophysical parameters. The precision of the classification could be extremely important to predict events, such earthquakes. In this sense, denoising de seismogram with statistical methods is crucial for precision of such wave type classifications and naturally, wavelet based statistical methods are some of the most attractive and proposed methods for this purpose. Chik et al. (2009), To et al. (2009), Ansari et al. (2010), Beenamol et al. (2012, 2016), Mousavi et al. (2016) and Vargas and Veiga (2017) are some of relevant works related to wavelet based methods applied in denoising seismic data.  

We applied the proposed shrinkage rule in denoising seismic amplitudes dataset collected from Coso, California and available in RSEIS R package (Lees et al., 2020). The available dataset has 726 seismic amplitudes measured in a short time interval. We considered $J = \lfloor \log_{2}726 \rfloor = 9$ resolution levels, thus $n = 2^9 = 512$ data points. The considered seismogram is shown in Figure \ref{fig:seismic}(a). For more details of Coso geological studies and the dataset, the reader is addressed to Lees (2004).

After application of a DWT using Daub10 basis, we observed an empirical right asymmetry of the coefficients, with $\hat{Skew}(d) = 0,25$ and $\hat{\sigma}=298,38$. Then, our asymmetric shrinkage rule was applied for denoising the empirical coefficients, with hyperparameters $a=2$, $b=3$, $\alpha = \alpha(j)$ and $m = m(j)$ according to \eqref{eq:alpha} and \eqref{eq:m} respectively. The denoised seismogram is shown in Figure \ref{fig:seismic}(b). From comparison of figures in \ref{fig:seismic}, one can observe considered noise reduction, mainly on the final period of the seismogram, when the seismic activity becomes to decrease. 

\begin{figure}[H]
\centering
\subfigure[Seismic dataset]{
\includegraphics[scale=0.5]{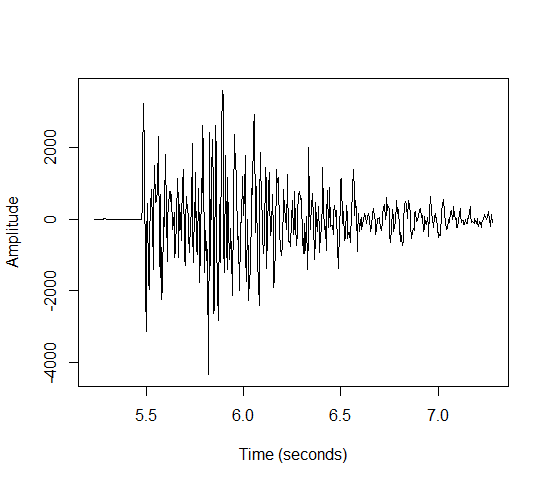}}
\subfigure[Denoised data]{
\includegraphics[scale=0.5]{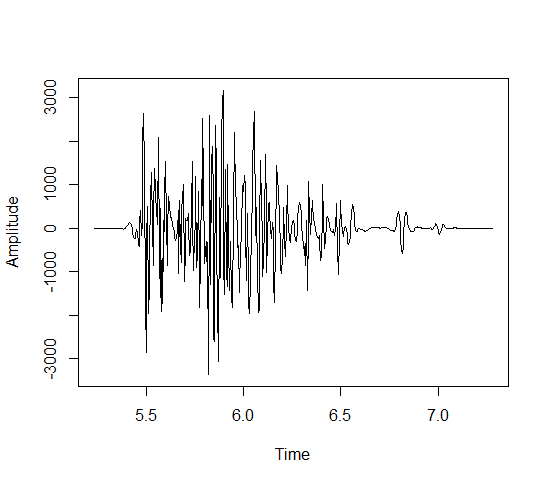}}
\caption{Considered seismogram from Coso, California (a) and its denoised version (b) after application of the proposed asymmetric shrinkage rule under beta prior.} \label{fig:seismic}
\end{figure}

Noise reduction can also be observed directly from the wavelet coefficients of the considered seismogram and its denoised version, shown in Figures \ref{fig:seismiccoef} (a) and (b) respectively. In fact, denoising occures mainly at high resolution levels, where most of non-zero magnitudes of the coefficients are tipically attributed to noise, as already mentioned in Section 5.  

\begin{figure}[H]
\centering
\subfigure[Empirical coefficients.]{
\includegraphics[scale=0.5]{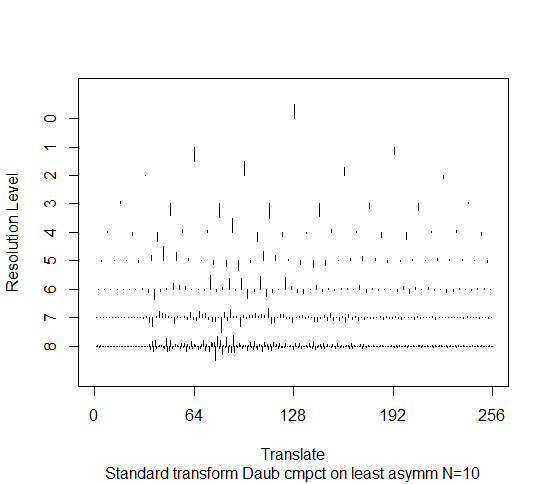}}
\subfigure[Shrunk coefficients.]{
\includegraphics[scale=0.5]{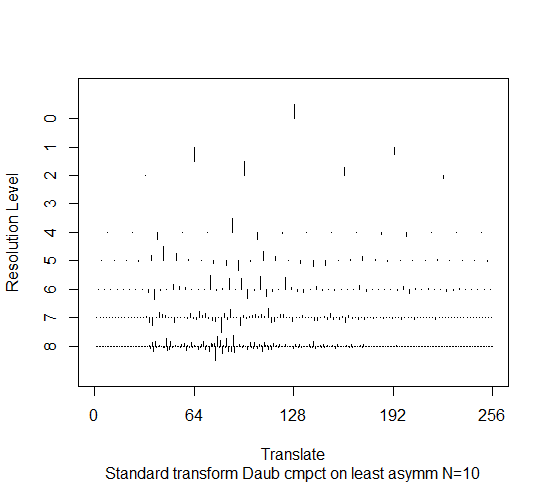}}
\caption{Empirical wavelet coefficients of the considered seismogram dataset (a) and their shrunk versions after denoising by the proposed asymmetric shrinkage rule under beta prior (b).} \label{fig:seismiccoef}
\end{figure}

Finally, we present the asymmetric shrinkage process in Figures \ref{fig:estcoef} (a) and (b). The first one shows us the plot of empirical coefficients against their shrunk versions ones. There is a weak right asymmetry around zero, i.e, empirical coefficients greater than zero had a little bit stronger shrinkage than the smaller ones, which is the effect of the chosen hyperparameters of the beta, $a=2$ and $b=3$. Although the asymmetry is not high, the adaptive asymmetric prior assigned according to this weak deviation of symmetry improved the denoising performance of the shrinkage rule. Figure \ref{fig:estcoef} (b) presents the histogram of shrunk coefficients, which emphasizes the sparsity of the estimated coefficients vector, with more than a half of the 512 coefficients shrunk to zero or very close to it.   

\begin{figure}[H]
\centering
\subfigure[Empirical against shrunk coefficients.]{
\includegraphics[scale=0.5]{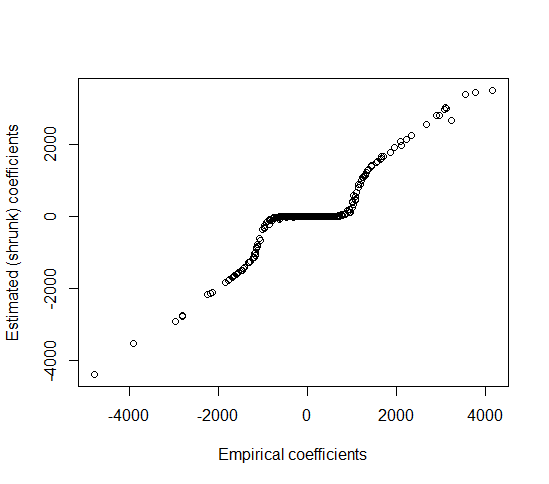}}
\subfigure[Histogram of shrunk coefficients.]{
\includegraphics[scale=0.5]{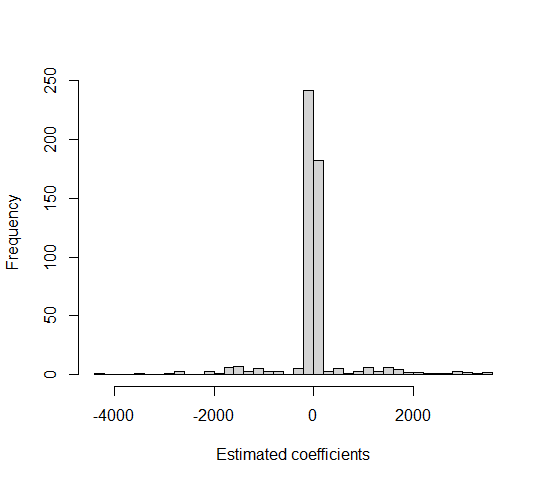}}
\caption{Empirical coefficients against their shrunk versions (a) and histogram of shrunk coefficients obtained by application of asymmetric shrinkage rule under beta prior.} \label{fig:estcoef}
\end{figure}

\section{Conclusion}
We propose in this work the use of asymmetric prior based on the beta distribution to the wavelet coefficients, which is a novelty in the wavelet shrinkage, since all the already proposed shrinkage methods are typically symmetric around zero. Moreover, few studies are concerned to bounded energy signals, which imply in bounded wavelet coefficients. In this sense, a bounded prior proposition, as the beta distribution, can be an alternative for applications in this context.   

The easy interpretation of the beta hyperparameters $a$ and $b$ in terms of asymmetry and shrinkage level and the well known flexibility of this distribution allow elicitation of the hyperparameters and adaptivity in modelling the coefficients, which are very attractive in bayesian setup. Further, the associated shrinkage rule had great performance in simulation studies and outperformed the considered shrinkage/thresholding methods in most of the scenarios, mainly when coefficients are highly asymmetric. Even when the coefficients are close to symmetry, the asymmetric shrinkage rule showed satisfactory results. These features allow the asymmetric beta to be considered by practitioners as a candidate to bayesian modelling of wavelet coefficients. 

The impact of wavelet basis choice, the proposition of other asymmetric distributions to wavelet coefficients and the evaluation of the proposed beta shrinkage rule in other performance measures are trivial possible extensions of this work.       

\section{Acknowledgement}
This study was financed by the Coordenação de Aperfeiçoamento de Pessoal de Nível Superior – Brasil (CAPES).

\end{document}